\documentclass[11pt]{article}

\usepackage[paper=a4paper]{geometry}
\usepackage{amsmath, amsfonts, amssymb, amsthm}
\usepackage{xspace}
\usepackage{tikz}
\usetikzlibrary{calc}

\usepackage[firstinits,style=alphabetic,maxnames=99,]{biblatex}
\AtEveryBibitem{%
    \ifentrytype{inproceedings}{%
        \clearfield{year}%
    }{}
}

\renewcommand*{\multicitedelim}{\addcomma\space} 

\bibliography{longnames.bib, bib-latest.bib}

\renewcommand{\geq}{\geqslant}
\renewcommand{\leq}{\leqslant}
\renewcommand{\epsilon}{\varepsilon}
\newcommand{\calT}{\mathcal{T}}
\newcommand{\bbZ}{\mathbb{Z}}

\newcommand{\expect}[1]{\mathbb{E}\left[{#1}\right]}
\newcommand{\expectinline}[1]{\mathbb{E}[{#1}]}
\newcommand{\twodots}{\ensuremath{\,.\,.\,}}

\newcommand{\VA}{\ensuremath{U}}      
\newcommand{\VB}{\ensuremath{V}}      
\newcommand{\DE}{\ensuremath{\Delta}} 
\newcommand{\IT}[1]{\ensuremath{IT(#1})}
\newcommand{\ANS}{\ensuremath{P}}
\newcommand{\ANSint}[1]{\ensuremath{\ANS_{[{#1}]}}}
\newcommand{\Dint}[1]{\ensuremath{\DE_{[{#1}]}}}
\newcommand{\DNint}[1]{\ensuremath{\DE_{[{#1}]^\textup{c}}}}
\newcommand{\Dfix}{\ensuremath{\DE^{\textup{fix}}_{[t_0,t_1]^\textup{c}}}}
\newcommand{\update}{\ensuremath{\textup{next}}}
\newcommand{\MVL}{\ensuremath{M_{V,\ell}}}

\newcommand{\NPtime}{\ensuremath{\mathbf{NP}}\xspace}

\newcommand{\Xc}{\ensuremath{X[t_0,t_1]^\textup{c}}}
\newcommand{\Xfix}{\ensuremath{X^{\textup{fix}}{[t_0,t_1]^\textup{c}}}}

\newcommand{\Patrascu}{P{\v a}tra{\c s}cu\xspace}

\newtheorem{theorem}{Theorem}
\newtheorem{lemma}[theorem]{Lemma}
\newtheorem{corollary}[theorem]{Corollary}

\newtheorem{fact}[theorem]{Fact}
\newtheorem{problem}[theorem]{Problem}

\newcommand{\inserttreefigure}{
\newcommand{\nodefont}{\scriptsize}
\begin{figure}[t]
\centering
\begin{tikzpicture}[level distance=3.9mm]
\tikzstyle{every node}=[fill=black!100,circle,inner sep=1pt]
\tikzstyle{level 1}=[sibling distance=72.5mm]
\tikzstyle{level 2}=[sibling distance=36mm]
\tikzstyle{level 3}=[sibling distance=18.2mm]
\tikzstyle{level 4}=[sibling distance=9.2mm]
\tikzstyle{level 5}=[sibling distance=4.5mm,set style={{every node}+=[fill=white!15]}]
\node {}
    child {node {}
        child {node {}
            child {node {}
                child {node {}
                    child {node {\nodefont 0}}
                    child {node {\nodefont 1}}
                }
                child {node {}
                    child {node {\nodefont 2}
                    }
                    child {node {\nodefont 3}
                    }
                }
            }
            child {node {}
                child {node {}
                    child {node {\nodefont 4}}
                    child {node {\nodefont 5}}
                }
                child {node {}
                    child {node {\nodefont 6}
                    }
                    child {node {\nodefont 7}
                    }
                }
            }
        }
        child {node {}
            child {node {}
                child {node {}
                    child {node {\nodefont 8}}
                    child {node {\nodefont 9}}
                }
                child {node {}
                    child {node {\nodefont 10}
                    }
                    child {node {\nodefont 11}
                    }
                }
            }
            child {node {}
                child {node {}
                    child {node {\nodefont 12}}
                    child {node {\nodefont 13}}
                }
                child {node {}
                    child {node {\nodefont 14}
                    }
                    child {node {\nodefont 15}
                    }
                }
            }
        }
    }
    child {node[fill=black!00,circle,draw,inner sep=1.5pt] {$v$}
        child {node {}
            child {node {}
                child {node {}
                    child {node {\nodefont 16}}
                    child {node {\nodefont 17}}
                }
                child {node {}
                    child {node {\nodefont 18}
                    }
                    child {node {\nodefont 19}
                    }
                }
            }
            child {node {}
                child {node {}
                    child {node {\nodefont 20}}
                    child {node {\nodefont 21}}
                }
                child {node {}
                    child {node {\nodefont 22}
                    }
                    child {node {\nodefont 23}
                    }
                }
            }
        }
        child {node {}
            child {node {}
                child {node {}
                    child {node {\nodefont 24}}
                    child {node {\nodefont 25}}
                }
                child {node {}
                    child {node {\nodefont 26}
                    }
                    child {node {\nodefont 27}
                    }
                }
            }
            child {node {}
                child {node {}
                    child {node {\nodefont 28}}
                    child {node {\nodefont 29}}
                }
                child {node {}
                    child {node {\nodefont 30}
                    }
                    child {node {\nodefont 31}
                    }
                }
            }
        }
    };
\end{tikzpicture}
\caption{\label{fig:tree}A lower bound tree $\calT$ over $n=32$ operations.}
\end{figure}
}

\newcommand{\insertmultfigure}{
\begin{figure}[t]
    \centering
    \small
    \makebox[3cm]{}%
    \makebox[8.5cm]{$\xleftarrow{\hspace*{3.5cm}}$ \hfill $n$ \hfill $\xrightarrow{\hspace*{3.5cm}}$}%
    \makebox[1cm]{}\\
    \vspace{0.2cm}
    \makebox[2cm]{}%
    \makebox[4cm]{}%
    \makebox[1.5cm]{$\xleftarrow{\hspace*{0.3cm}}$ \hfill $\ell$ \hfill $\xrightarrow{\hspace*{0.3cm}}$}%
    \makebox[4cm]{$\xleftarrow{\hspace*{1.3cm}}$ \hfill $t_0$ \hfill $\xrightarrow{\hspace*{1.3cm}}$}%
    \makebox[1cm]{}\\
    \vspace{0.1cm}
    \makebox[3cm]{}%
    \framebox[3cm]{\phantom{$X'$}}%
    \framebox[1.5cm]{$X'$}%
    \framebox[4cm]{\phantom{$X'$}}%
    \makebox[1cm]{$=X$}\\
    \vspace{0.4cm}
    \makebox[2cm]{}%
    \makebox[6.5cm]{}%
    \makebox[3cm]{$\xleftarrow{\hspace*{0.8cm}}$ \hfill $2\ell$ \hfill $\xrightarrow{\hspace*{0.8cm}}$}%
    \makebox[1cm]{}\\
    \vspace{0.1cm}
    \makebox[3cm]{}%
    \framebox[5.5cm]{\phantom{$Y'$}}%
    \framebox[3cm]{$Y'$}%
    \makebox[1cm]{$=Y$}\\
    \vspace{0.4cm}
    \makebox[2cm]{}%
    \makebox[2.5cm]{}%
    \makebox[1.5cm]{$\xleftarrow{\hspace*{0.3cm}}$ \hfill $\ell$ \hfill $\xrightarrow{\hspace*{0.3cm}}$}%
    \makebox[5.5cm]{$\xleftarrow{\hspace*{1.7cm}}$ \hfill $t_0+\ell$ \hfill $\xrightarrow{\hspace*{1.7cm}}$}%
    \makebox[1cm]{}\\
    \vspace{0.1cm}
    \framebox[4.5cm]{\phantom{$Z'$}}%
    \framebox[1.5cm]{$Z'$}%
    \framebox[5.5cm]{\phantom{$Z'$}}%
    \makebox[1cm]{$=Z$}\\
    ~
    \caption{\label{fig:retrorse}$X$, $Y$ and $Z=X\times Y$ in base $q$.}
\end{figure}
}

\title{Tight Cell-Probe Bounds for Online Integer Multiplication and Convolution\thanks{~A preliminary version of this paper appeared in ICALP '11.}}

\author{
    Rapha\"{e}l Clifford%
    \thanks{~University of Bristol, Department of Computer Science, Bristol, U.K.}
    \and Markus Jalsenius\footnotemark[2]%
}

\date{}

\begin{document}

\maketitle

\begin{abstract}
    We show tight bounds for both online integer multiplication and convolution in the cell-probe model with word size $w$. For the multiplication problem, one pair of digits, each from one of two $n$ digit numbers that are to be multiplied, is given as input at step $i$. The online algorithm outputs a single new digit from the product of the numbers before step $i+1$. We give a $\Theta{\left(\frac{\delta}{w} \log n\right)}$ bound on average per output digit for this problem where $2^\delta$ is the maximum value of a digit. In the convolution problem, we are given a fixed vector $V$ of length $n$ and we consider a stream in which numbers arrive one at a time. We output the inner product of $V$ and the vector that consists of the last $n$ numbers of the stream. We show a $\Theta{\left(\frac{\delta}{w}\log n\right)}$ bound for the number of probes required per new number in the stream. All the bounds presented hold under randomisation and amortisation. Multiplication and convolution are central problems in the study of algorithms which also have the widest range of practical applications.
\end{abstract}



\section{Introduction}

We consider two related and fundamental problems: multiplying two integers and computing the convolution or cross-correlation of two vectors. We study both these problems in an online or streaming context and provide matching upper and lower bounds in the cell-probe model.  The importance of these problems is hard to overstate with both the integer multiplication and convolution problems playing a central role in modern algorithms design and theory.

For notational brevity, we write $[q]$ to denote the set $\{0,\dots,q-1\}$, where $q$ is a positive integer.

\begin{problem}[Online convolution]
    For a fixed vector $V\in[q]^n$ of length $n$, we consider a stream in which numbers from $[q]$ arrive one at a time. For each arriving number, before the next number arrives, we output the inner product (modulo~$q$) of $V$ and the vector that consists of the last $n$ numbers of the stream.
\end{problem}

We show that there are instances of this problem such that any algorithm solving it will require $\Omega(\frac{\delta}{w}\log n)$ amortised time on average per output, where $\delta=\log_2 q$ and $w$ is the number of bits per cell in the cell-probe model. The result is formally stated in Theorem~\ref{thm:conv}.

\begin{problem}[Online multiplication]
    Given two numbers $X,Y\in [q^n]$, where $q$ is the base and $n$ is the number of digits per number, we want to output the $n$ least significant digits of the product of $X$ and $Y$, in base $q$. We must do this under the constraint that the $i$th digit of the product (starting from the lower-order end) is outputted before the $(i+1)$th digit, and when the $i$th digit is outputted, we only have access to the $i$ least significant digits of $X$ and $Y$,\! respectively. We can think of the digits of $X$ and $Y$ arriving online in pairs, one digit from each of $X$ and $Y$.
\end{problem}

We show that there are instances of this problem such that any algorithm solving it takes $\Omega(\frac{\delta}{w} \log n)$ time on average per input pair, where $\delta=\log_2 q$ and $w$ is the number of bits per cell in the cell-probe model. The result is formally stated in Theorem~\ref{thm:mult}

Our main technical innovation is to extend recently developed methods designed to give lower bounds on dynamic data structures to the seemingly distinct field of online algorithms.  Where $\delta = w$, for example, we have $\Omega(\log{n})$ lower bounds for both online multiplication and convolution, thereby matching the currently best known offline upper bounds in the RAM model. As we discuss in the Section~\ref{sec:previous}, this may be the highest lower bound that can be formally proved for these problems without a further significant theoretical breakthrough.

For the convolution problem, one consequence of our results is a new separation between the time complexity of exact and inexact string matching in a stream. The convolution has played a particularly important role in the field of combinatorial pattern matching where many of the fastest algorithms rely crucially for their speed on the use of fast Fourier transforms (FFTs) to perform repeated convolutions. These methods have also been extended to  allow searching for patterns in rapidly processed data streams~\cite{CEPP:2011,CS:2011}. The results we present here therefore give the first strict separation between the constant time complexity of online exact matching~\cite{Galil:1981} and any convolution based online pattern matching algorithm.

Although we show only the existence of probability distributions on the inputs for which we can prove lower bounds on the expected running time of any deterministic algorithm, by Yao's minimax principle~\cite{Yao1977:Minimax} this also immediately implies that for every (randomised) algorithm, there is a worst-case input such that the (expected) running time is equally high.  Therefore our lower bounds hold equally for randomised algorithms as for deterministic ones.

The lower bounds we show for both online multiplication and convolution are also tight within the cell-probe model. This can be seen by application of reductions described in~\cite{FS:1973, CEPP:2011}.  It was shown there that any offline algorithm for multiplication~\cite{FS:1973} or convolution~\cite{CEPP:2011} can be converted to an online one with at most an $O(\log{n})$ factor overhead. For details of these reductions we refer the reader to the original papers.  In our case, the same approach also allows us to directly convert any cell-probe algorithm from an offline to online setting. An offline cell-probe algorithm for either multiplication or convolution could first read the whole input, then compute the answers and finally output them. This takes $O{(\frac{\delta}{w} n)}$ cell probes. We can therefore derive online cell-probe algorithms which take only $O{(\frac{\delta}{w}\log n)}$ probes per output, hence matching the new lower bound we give.

\subsection{Previous results and upper bounds in the RAM model} \label{sec:previous}

The best time complexity lower bounds for online multiplication of two $n$-bit numbers were given in the 1974 by Paterson, Fischer and Meyer.  They presented an $\Omega(\log{n})$ lower bound for multitape Turing machines~\cite{PFM:1974} and also gave an $\Omega(\log{n}/\log{\log n})$ lower bound for the `bounded activity machine' (BAM).  The BAM, which is a strict generalisation of the Turing machine model but which has nonetheless largely fallen out of favour, attempts to capture the idea that future states can only depend on a limited part of the current configuration.  To the authors' knowledge, there has been no progress on cell-probe lower bounds for online multiplication or convolution previous to the work we present here.

There have however been attempts to provide offline lower bounds for the related problem of computing the FFT.  In~\cite{Morgenstern:1973} Morgenstern gave an $\Omega(n \log{n})$ lower bound conditional on the assumption that the underlying field of the transform is the complex numbers and that the modulus of any complex numbers involved in the computation is at most $1$.  Papadimitriou gave the same $\Omega(n \log{n})$ lower bound for FFTs of length a power of two, this time excluding certain classes of algorithms including those that rely on linear mathematical relations among the roots of unity~\cite{Papadimitriou:1979}.  This work had the advantage of giving a conditional lower bound for FFTs over more general algebras than was previously possible, including for example finite fields.  In 1986 Pan~\cite{Pan:1986} showed that another class of algorithms having a so-called synchronous structure must require $\Omega(n \log{n})$ time for the computation of both the FFT and convolution.

The fastest known algorithms for both offline integer multiplication and convolution in the word-RAM model require $O(n\log{n})$ time by a well known application of a constant number of FFTs.  As a consequence our online lower bounds match the best known time upper bounds for the offline problem. As we discussed above, our lower bounds are also tight within the cell-probe model for the online problems.   The question now naturally arises as to whether one can find higher lower bounds in the RAM model.   This appears as an interesting question as there remains a gap between the best known time upper bounds provided by existing algorithms and the lower bounds that we give within the cell-probe model.  However, as we mention above, any offline algorithm for convolution or multiplication can be converted to an online one with at most an $O(\log{n})$ factor overhead~\cite{FS:1973,CEPP:2011}. As a consequence, it is likely to be hard to prove a higher lower bound for the online problem than we have given, at least for the case where ${\delta}/w \in \Theta(1)$, as this would immediately imply a superlinear lower bound for offline convolution or multiplication.  Such superlinear lower bounds are not yet known for any problem in \NPtime except in very restricted models of computation, such as for example a single tape Turing Machine.  Our only alternative route to find tight time bounds would be to find better upper bounds for the online problems. For the case of online multiplication at least, this has been an open problem since at least 1973 and has so far resisted our best attempts.

\subsection{The cell-probe model}

When stating lower bounds it is important to be precise about the model in which the bounds apply. Our bounds in this paper hold in perhaps the strongest model of them all, the \emph{cell-probe model}, introduced originally by Minsky and Papert~\cite{MP:1969,MP:1988} in a different context and then subsequently by Fredman~\cite{Fredman:1978} and Yao~\cite{Yao1981:Tables}. In this model, there is a separation between the computing unit and the memory, which is external and consists of a set of cells of $w$ bits each. The computing unit cannot remember any information between operations. Computation is free and the cost is measured only in the number of cell reads or writes (cell-probes). This general view makes the model very strong, subsuming for instance the popular word-RAM model. In the word-RAM model certain operations on words, such as addition, subtraction and possibly multiplication take constant time (see for example~\cite{Hagerup:1998} for a detailed introduction). Here a word corresponds to a cell. Typically we think of the cell size $w$ as being at least $\log_2 n$ bits, where $n$ is the number of cells. This allows each cell to hold the address of any location in memory.

The generality of the cell-probe model makes it particularly attractive for establishing lower bounds for data structure problems and many such results have been given in the past couple of decades. The approaches taken have until recently mainly been based on communication complexity arguments and the chronogram technique of Fredman and Saks~\cite{FS1989:chronogram}. There remains however, a number of unsatisfying gaps between the lower bounds and known upper bounds. Only a few years ago, a breakthrough lead by Demaine and \Patrascu gave us the tools to seal the gaps for several data structure problems~\cite{DP2006:Low-Bounds}. The new technique was based on information theoretic arguments. Demaine and \Patrascu also presented ideas which allowed them to express more refined lower bounds such as trade-offs between updates and queries of dynamic data structures.  For a list of data structure problems and their lower bounds using these and related techniques, see for example~\cite{Pat2008:Thesis}.

\subsection{Organisation}

We present the new cell-probe lower bound for online convolution in Section~\ref{sec:convolution} along with the main techniques that we will use throughout.  In Section~\ref{sec:multiplication} we show how these can then be applied to the problem of online multiplication.

\section{Online convolution} \label{sec:convolution}

For a vector $V$ of length $n$ and $i\in [n]$, we write $V[i]$ to denote the elements of $V$. For positive integers $n$ and $q$, the \emph{inner product} of two vectors $\VA,\VB\in [q]^n$, denoted $\langle\VA,\VB\rangle$, is defined as
\begin{equation*}
    \langle \VA,\VB\rangle = \sum_{i\in [n]} (\VA[i]\cdot \VB[i])\,.
\end{equation*}

Parameterised by two positive integers $n$ and $q$, and a fixed vector $\VB\in [q]^n$, the \emph{online convolution problem} asks to maintain a vector $\VA\in [q]^n$ subject to an operation $\update(\DE)$, which takes a parameter $\DE\in [q]$, modifies $\VA$ to be the vector $(\VA[1],\VA[2],\dots,\VA[n-1],\DE)$ and then returns the inner product $\langle \VA,\VB\rangle$. In other words,  $\update(\DE)$ modifies $\VA$ by shifting all elements one step to the left, pushing the leftmost element out, and setting the new rightmost element to $\DE$. We consider the online convolution problem over the ring $\bbZ/q\bbZ$, that is integer arithmetic modulo~$q$. Let $\delta = \log_2 q$.

\begin{theorem}
    \label{thm:conv}
      For any positive integers $q$ and $n$, in the cell probe model with $w$ bits per cell there exist instances of the online convolution problem such that the expected amortised time per $\update$-operation is $\Omega\left(\frac{\delta}{w}\log n\right)$, where $\delta=\log_2 q$.
\end{theorem}

In order to prove Theorem~\ref{thm:conv} we will consider a random instance that is described by $n$ $\update$-operations on the sequence $\DE=(\DE_0,\dots,\DE_{n-1})$, where each $\DE_i$ is chosen independently and uniformly at random from $[q]$. We defer the choice of the fixed vector $\VB$ until later. For $t$ from 0 to $n-1$, we use $t$ to denote the time, and we say that the operation $\update(\DE_t)$ occurs at time $t$.

We may assume that prior to the first update, the vector $\VA=\{0\}^n$, although any values are possible since they do not influence the analysis. To avoid technicalities we will from now on assume that $n$ is a power of two.

\subsection{Information transfer}

Following the overall approach of Demaine and \Patrascu~\cite{DP2004:Partial-sums} we will consider adjacent time intervals and study the \emph{information} that is transferred from the operations in one interval to the next interval. More precisely, let $t_0,t_1,t_2\in [n]$ such that $t_0\leq t_1< t_2$ and consider any algorithm solving the online convolution problem. We would like to keep track of the memory cells that are written to during the time interval $[t_0,t_1]$ and then read during the succeeding interval $[t_1+1,t_2]$. The information from the \update-operations taking place in the interval $[t_0,t_1]$ that the algorithm passes on to the interval $[t_1+1,t_2]$ must be contained in these cells. Informally one can say that there is no other way for the algorithm to determine what occurred during the interval $[t_0,t_1]$ except through these cells. Formally, the \emph{information transfer}, denoted $IT(t_0,t_1,t_2)$, is defined to be the set of memory cells $c$ such that $c$ is written during $[t_0,t_1]$, read at a time $t_r\in [t_1+1,t_2]$ and not written during $[t_1+1,t_r]$. Hence a cell that is overwritten in $[t_1+1,t_2]$ before being read is not included in the information transfer. Observe that the information transfer depends on the algorithm, the vector $\VB$ and the sequence $\DE$. The first aim is to show that for any choice of algorithm solving the online convolution problem, the number of cells in the information transfer is bounded from below by a sufficiently large number for some choice of the vector~$\VB$.

For $0\leq t_0\leq t_1< n$, we write $\Dint{t_0,t_1}$ to denote the subsequence $(\DE_{t_0},\dots,\DE_{t_1})$ of $\DE$, and $\DNint{t_0,t_1}$ to denote the sequence $(\DE_{0},\dots,\DE_{t_0-1},\DE_{t_1+1},\dots,\DE_{n-1})$ which contains all the elements of $\DE$ except for those in $\Dint{t_0,t_1}$. For $t\in [n]$, we let $\ANS_t\in [q]$ denote the inner product returned by $\update(\DE_t)$ at time $t$ (recall that we operate modulo $q$). We let $\ANSint{t_1+1,t_2} = (\ANS_{t_1+1},\dots,\ANS_{t_2})$.

Since $\DE$ is a random variable, so is $\ANSint{t_1+1,t_2}$. In particular, if we condition on a fixed choice of $\DNint{t_0,t_1}$, call it $\Dfix$, then $\ANSint{t_1+1,t_2}$ is a random variable that depends on the random values in $\Dint{t_0,t_1}$. The dependency on the $\update$-operations in the interval $[t_0,t_1]$ is captured by the information transfer $IT(t_0,t_1,t_2)$, which must encode all the relevant information in order for the algorithm to correctly output the inner products in $[t_1+1,t_2]$. In other words, an encoding of the information supplied by cells in the information transfer is an upper bound on the conditional \emph{entropy} of $\ANSint{t_1+1,t_2}$. This fact is stated in Lemma~\ref{lem:entropy-upper} and was given in~\cite{Pat2008:Thesis} with small notational differences.


\begin{lemma}[Lemma~3.2 of \cite{Pat2008:Thesis}]
    \label{lem:entropy-upper}
    The entropy
    \begin{align*}
        H(\ANSint{t_1+1,t_2} \,&\mid\, \DNint{t_0,t_1} = \Dfix) \leq\\
            & w + 2w\cdot \expect{|\IT{t_0,t_1,t_2}| \,\mid\, \DNint{t_0,t_1} = \Dfix}\,.
    \end{align*}
\end{lemma}
\begin{proof}
The average length of any encoding of $\ANSint{t_1+1,t_2}$ (conditioned on $\Dfix$) is an upper bound on its entropy. We use the information transfer as an encoding in the following way. For every cell $c$ in the information transfer $\IT{t_0,t_1,t_2}$, we store the address of $c$, which takes at most $w$ bits under the assumption that the cell size can hold the address of every cell, and we store the contents of $c$, which is a cell of $w$ bits. In total this requires $2w\cdot |\IT{t_0,t_1,t_2}|$ bits. In addition, we store the size of the information transfer, $|\IT{t_0,t_1,t_2}|$, so that any algorithm decoding the stored information knows how many cells are stored and hence when to stop checking for stored cells. Storing the size of the information transfer requires $w$ bits, thus the average total length of the encoding is $w + 2w\cdot \expectinline{|\IT{t_0,t_1,t_2}| \,\mid\, \DNint{t_0,t_1} = \Dfix}$.

In order to prove that the described encoding is valid, we describe how to decode the stored information. We do this by simulating the algorithm. First we simulate the algorithm from time 0 to $t_0-1$. We have no problem doing so since all necessary information is available in $\Dfix$, which we know. We then skip from time $t_0$ to $t_1$ and resume simulating the algorithm from time $t_1+1$ to $t_2$. In this interval, the algorithm outputs the values in $\ANSint{t_1+1,t_2}$. In order to correctly do so, the algorithm might need information from the $\update$-operations in $[t_0,t_1]$. This information is only available through the encoding described above. When simulating the algorithm, for each cell $c$ we read, we check if the address of $c$ is contained in the list of addresses that was stored. If so, we obtain the contents of $c$ by reading its stored value. Each time we write to a cell whose address is in the list of stored addresses, we remove it from the stored list, or blank it out. Note that every cell we read whose address is not in the stored list contains a value that was written last either before time $t_0$ or after time $t_1$. Hence its value is known to us.
\end{proof}

\subsection{Recovering information} \label{sec:conv-recover}

In the previous section, we provided an upper bound for the entropy of the outputs from the $\update$-operations in $[t_1+1,t_2]$. Next we will explore how much information \emph{needs to be} communicated from $[t_0,t_1]$ to $[t_1+1,t_2]$. This will provide a lower bound on the entropy. As we will see, the lower bound can be expressed as a function of the length of the intervals and the vector $\VB$.

Suppose that $[t_0,t_1]$ and $[t_1+1,t_2]$ both have the same length $\ell$. That is, $t_1-t_0+1=t_2-t_1=\ell$. For $i\in [\ell]$, the output at time $t_1+1+i$ can be broken into two sums $S_i$ and $S'_i$, such that $P_{t_1+1+i} = S_i + S'_i$, where
\begin{equation*}
    S_i = \sum_{j\in[\ell]} (\VB[n-1-(\ell+i)+j]\cdot \Delta_{t_0+j})
\end{equation*}
is the contribution from the alignment of $\VB$ with $\Dint{t_0,t_1}$, and $S'_i$ is the contribution from the alignments that do not include $\Dint{t_0,t_1}$.

We define $\MVL$ to be the $\ell$$\times$$\ell$ matrix with entries $\MVL(i,j)= \VB[n-1-(\ell+i)+j]$. That is,
\begin{equation*}
    \MVL =
        \begin{pmatrix}
            \VB[n-\ell-1] & \VB[n-\ell+0] & \VB[n-\ell+1] & \cdots & \VB[n-2] \\
            \VB[n-\ell-2] & \VB[n-\ell-1] & \VB[n-\ell+0] & \cdots & \VB[n-3] \\
            \VB[n-\ell-3] & \VB[n-\ell-2] & \VB[n-\ell-1] & \cdots & \VB[n-4] \\
            \vdots  & \vdots & \vdots & \ddots & \vdots  \\
            \VB[n-2\ell] & \VB[n-2\ell+1] & \VB[n-2\ell+2] & \cdots & \VB[n-\ell-1]
        \end{pmatrix}\,.
\end{equation*}
We observe that $\MVL$ is a \emph{Toeplitz} matrix (or ``upside down'' \emph{Hankel} matrix) since it is constant on each descending diagonal from left to right. This property will be important later. From the definitions above it follows that
\begin{equation}
    \label{eq:matrix}
    \MVL\times
        \begin{pmatrix}
            \DE_{t_0} \\
            \DE_{t_0+1} \\
            \vdots \\
            \DE_{t_1}
        \end{pmatrix}
        =
        \begin{pmatrix}
            S_0 \\
            S_1 \\
            \vdots \\
            S_{\ell-1}
        \end{pmatrix}\,.
\end{equation}

We define the \emph{recovery number} $R_{\VB, \ell}$ to be the number of variables $x\in \{x_1,\dots,x_\ell\}$ such that $x$ can be determined uniquely by the system of linear equations
\begin{equation*}
    \MVL\times
        \begin{pmatrix}
            x_1 \\
            \vdots \\
            x_\ell
        \end{pmatrix}
        =
        \begin{pmatrix}
            y_1 \\
            \vdots \\
            y_\ell
        \end{pmatrix}\,,
\end{equation*}
where we operate in $\bbZ/q\bbZ$. The recovery number may be distinct from the rank of a matrix, even where we operate over a field. As an example, consider the all ones matrix. The matrix will have recovery number zero but rank one.  The recovery number is however related to the conditional entropy of $\ANSint{t_1+1,t_2}$ as described by the next lemma.

\begin{lemma}
    \label{lem:entropy-lower}
    If the intervals $[t_0,t_1]$ and $[t_1+1,t_2]$ both have the same length~$\ell$, then the entropy
    \begin{equation*}
        H(\ANSint{t_1+1,t_2} \,\mid\, \DNint{t_0,t_1} = \Dfix)
            \geq \delta R_{\VB, \ell}\,.
    \end{equation*}
\end{lemma}
\begin{proof}
    As described above, for $i\in [\ell]$, $P_{t_1+1+i} = S_i + S'_i$, where $S'_i$ is a constant that only depends on $\VB$ and $\Dfix$. Hence we can compute the values $S_0,\dots,S_{\ell-1}$ from $\ANSint{t_1+1,t_2}$. From Equation~(\ref{eq:matrix}) it follows that $S_0,\dots,S_{\ell-1}$ uniquely specify $R_{\VB, \ell}$ of the parameters in $\Dint{t_0,t_1}$. That is, we can recover $R_{\VB, \ell}$ of the parameters from the interval $[t_0,t_1]$. Each of these parameters is a random variable that is uniformly distributed in $[q]$, so it contributes $\delta$ bits of entropy.
\end{proof}

We now combine Lemmas~\ref{lem:entropy-upper} and~\ref{lem:entropy-lower} in the following corollary.

\begin{corollary}
    \label{cor:conv-IT}
    For any fixed vector $\VB$, two intervals $[t_0,t_1]$ and $[t_1+1,t_2]$ of the same length $\ell$, and any algorithm solving the online convolution problem on $\DE$ chosen uniformly at random from $[q]^n$,
    \begin{equation*}
        \expect{|\IT{t_0,t_1,t_2}|}\geq \frac{\delta R_{\VB, \ell}}{2w}-\frac{1}{2}\,.
    \end{equation*}
\end{corollary}
\begin{proof}
    For $\DNint{t_0,t_1}$ fixed to $\Dfix$, comparing Lemmas~\ref{lem:entropy-upper} and~\ref{lem:entropy-lower}, we see that
    \begin{equation*}
        \expect{|\IT{t_0,t_1,t_2}| \,\mid\, \DNint{t_0,t_1} = \Dfix} \geq \frac{\delta R_{\VB, \ell}}{2w}-\frac{1}{2}\,.
    \end{equation*}
    The result follows by taking expectation over $\DNint{t_0,t_1}$ under the random sequence~$\Delta$.
\end{proof}

\subsection{The lower bound for online convolution} \label{sec:conv-lower}

We now show how a lower bound on the total number of cell reads over $n$ $\update$-operations can be obtained by summing the information transfer between many pairs of time intervals. We again follow the approach of Demaine and \Patrascu~\cite{DP2004:Partial-sums}, which involves conceptually constructing a balanced tree over the time axis. This \emph{lower bound tree}, denoted $\calT$, is a balanced binary tree on $n$ leaves. Recall that we have assumed that $n$ is a power of two. The leaves, from left to right, represent the time $t$ from $0$ to $n-1$, respectively. An internal node $v$ is associated with the times $t_0$, $t_1$ and $t_2$ such that the two intervals $[t_0,t_1]$ and $[t_1+1,t_2]$ span the left subtree and the right subtree of $v$, respectively. For example, in Figure~\ref{fig:tree},\inserttreefigure the node labelled $v$ is associated with the intervals $[16,23]$ and $[24,31]$.


For an internal node $v$ of $\calT$, we write $IT(v)$ to denote $IT(t_0,t_1,t_2)$, where $t_0$, $t_1$, $t_2$ are associated with $v$. We write $L(v)$ to denote the number of leaves in the left (same as the right) subtree of $v$. The key lemma, stated next, is a modified version of Theorem~3.6 in~\cite{Pat2008:Thesis}. The statement of the lemma is adapted to our online convolution problem and the proof relies on Corollary~\ref{cor:conv-IT}.

\begin{lemma}
    \label{lem:conv-time}
    For any fixed vector $\VB$ and any algorithm solving the online convolution problem, the expected running time of the algorithm over a sequence $\DE$ that is chosen uniformly at random from $[q]^n$ is at least
    \begin{equation*}
        \frac{\delta}{2w}\sum_{v\in \calT} R_{\VB,L(v)} - \frac{n-1}{2}\,,
    \end{equation*}
    where the sum is over the internal nodes of $\calT$.
\end{lemma}
\begin{proof}
    We first consider a fixed sequence $\DE$. We argue that the number of read instructions executed by the algorithm is at least $\sum_{v\in \calT}|IT(v)|$. To see this, for any read instruction, let $t_r$ be the time it is executed. Let $t_w\leq t_r$ be the time the cell was last written, ignoring $t_r=t_w$. Then this read instruction (the cell it acts upon), is contained in $IT(v)$, where $v$ is the lowest common ancestor of $t_w$ and $t_r$. Thus, $\sum_{v\in \calT}|IT(v)|$ never double-counts a read instruction.

    For a random $\DE$, an expected lower bound on the number of read instructions is therefore $\expectinline{\sum_{v\in \calT}|IT(v)|}$. Using linearity of expectation and Corollary~\ref{cor:conv-IT}, we obtain the lower bound in the statement of the lemma.
\end{proof}

\subsubsection{Lower bound with a random vector $\VB$} \label{sec:fixed}

We have seen in Lemma~\ref{lem:conv-time} that a lower bound is highly dependent on the recovery numbers of the vector $\VB$. In the next lemma, we show that a random vector $\VB$ has recovery numbers that are large.

\begin{lemma}
    \label{lem:invertible}
    Suppose that $q$ is a prime and the vector $\VB$ is chosen uniformly at random from $[q]^n$. Then $\expectinline{R_{\VB, \ell}}\geq \ell/2$ for every length $\ell$.
\end{lemma}
\begin{proof}
    Recall that $\MVL$ is an $\ell$$\times$$\ell$ Toeplitz matrix. It has been shown in~\cite{KL1996:Toeplitz} that for any $\ell$, out of all the $\ell$$\times$$\ell$ Toeplitz matrices over a finite field of $q$ elements, a fraction of exactly $(1-1/q)$ is non-singular. This fact was actually already established in~\cite{Day1960:Matrices} almost 40 years earlier but incidentally reproved in~\cite{KL1996:Toeplitz}. Since we have assumed in the statement of the lemma that $q$ is a prime, the ring $\bbZ/q\bbZ$ we operate in is indeed a finite field. The diagonals of $\MVL$ are independent and uniformly distributed in $[q]$, hence the probability that $\MVL$ is invertible is $(1-1/q)\geq 1/2$. If $\MVL$ is invertible then the recovery number $R_{\VB, \ell}=\ell$; there is a unique solution to the system of linear equations in Equation~(\ref{eq:matrix}). On the other hand, if $\MVL$ is not invertible then the recovery number will be lower. Thus, we can safely say that the expected recovery number $R_{\VB, \ell}$ is at least $\ell/2$, which proves the lemma.
\end{proof}

Before we give a lower bound for a random choice of $\VB$ in Theorem~\ref{thm:random-lowerbound} below, we state the following fact.

\begin{fact}
    \label{fac:leaves}
    For a balanced binary tree with $n$ leaves, the sum of the number of leaves in the subtree rooted at $v$, taken over all internal nodes $v$, is $n\log_2 n$.
\end{fact}

\begin{theorem}
    \label{thm:random-lowerbound}
    Suppose that $q$ is a prime. In the cell-probe model with $w$ bits per cell, any algorithm solving the online convolution problem on a vector $\VB$ and $\DE$, both chosen uniformly at random from $[q]^n$, will run in $\Omega\left(\frac{\delta}{w} n\log n\right)$ time in expectation, where $\delta=\log_2 q$.
\end{theorem}
\begin{proof}
    For a random vector $\VB$, a lower bound is obtained by taking the expectation of the bound in the statement of Lemma~\ref{lem:conv-time}. Using linearity of expectation and applying Lemma~\ref{lem:invertible} and Fact~\ref{fac:leaves} completes the proof.
\end{proof}

 \noindent
 \textit{Remark.}
      Theorem~\ref{thm:random-lowerbound} requires that $q$ is a prime but for an integer $\delta>1$, $q=2^\delta$ is not a prime. However,  we know that there is always at least one prime $p$ such that $2^{\delta-1} < p < 2^\delta$. Thus, Theorem~\ref{thm:random-lowerbound} is applicable for any integer $\delta$, only with an adjustment by at most one.

\subsubsection{Lower bound with a fixed vector $\VB$}

We demonstrate next that it is possible to design a fixed vector $\VB$ with guaranteed large recovery numbers. We will use this vector in the proof of Theorem~\ref{thm:conv}. The idea is to let $\VB$ consist of stretches of~0s interspersed by~1s. The distance between two succeeding 1s is an increasing power of two, ensuring that for half of the alignments in the interval $[t_1+1,t_2]$, all but exactly one element of $\Dint{t_0,t_1}$ are simultaneously aligned with a 0 in $\VB$. We define the binary vector $K_n\in[2]^n$ to be
\begin{equation*}
    K_n=(\dots\texttt{0000000000000{\bf 1}000000000000000{\bf 1}0000000{\bf 1}000{\bf 1}0{\bf 11}0})\,,
\end{equation*}
or formally,
\begin{equation}
    \label{eq:K}
    K_n[i] =
    \begin{cases}
        1, &\text{if $n-1-i$ is a power of two;}\\
        0, &\text{otherwise.}
    \end{cases}
\end{equation}

\begin{lemma}
    \label{lem:fixed}
    Suppose $\VB=K_n$ and $\ell\geq 1$ is a power of two. The recovery number $R_{\VB, \ell}\geq \ell/2$.
\end{lemma}
\begin{proof}
    Recall that entry $\MVL(i,j)= \VB[n-1-(\ell+i)+j]$. Thus, $\MVL(i,j)=1$ if and only if $n-1-(n-1-(\ell+i)+j)=\ell+i-j$ is a power of two. It follows that for row $i=\ell/2,\dots,\ell-1$, $\MVL(i,j)=1$ for $j=i$ and $\MVL(i,j)=0$ for $j\neq i$. This implies that the recovery number $R_{\VB, \ell}$ is at least $\ell/2$.
\end{proof}

We finally give the proof of Theorem~\ref{thm:conv}.

\begin{proof}[Theorem~\ref{thm:conv}]
    We assume that $n$ is a power of two. Let $\VB=K_n$. It follows from Lemma~\ref{lem:fixed} and Fact~\ref{fac:leaves} that $\sum_{v\in \calT} R_{\VB,L(v)}\geq \sum_{v\in \calT} L(v)/2 = \Omega\left(n \log n\right)$. Note that $L(v)$ is a power of two for every node $v$ in $\calT$. For $\Delta$ chosen uniformly at random from $[q]^n$, apply Lemma~\ref{lem:conv-time} to obtain the expected running time $\Omega\left(\frac{\delta}{w} n\log n\right)$ over $n$ $\update$-operations.
\end{proof}

\section{Online multiplication} \label{sec:multiplication}

In this section we consider online multiplication of two $n$-digit numbers in base $q\geq 2$. For a non-negative integer $X$, let $X[i]$ denote the $i$th digit of $X$ in base $q$, where the positions are numbered starting with 0 at the right (lower-order) end. We think of $X$ padded with zeros to make sure that $X[i]$ is defined for arbitrarily large $i$. For $j\geq i$, we write $X[i\twodots j]$ to denote the integer that is written $X[j]\cdots X[i]$ in base $q$.  For example, let $X=15949$ (decimal representation) and $q=8$ (octal):
\begin{align*}
    &X=37115 \textrm{ (base 8)}  &X[0]=5  \\
    &X[1\twodots 3]= 711 \textrm{ (base 8)} = 457 \textrm{ (decimal)}  &X[3]=7\\
    &X[3\twodots 10]= 37 \textrm{ (base 8)} = 31 \textrm{ (decimal)}  &X[15]=0
\end{align*}

The \emph{online multiplication problem} is defined as follows. The input is two $n$-digit numbers $X,Y \in [q^n]$ in base $q$ (higher order digits may be zero). Let $Z=X\times Y$. We want to output the $n$ lower order digits of $Z$ in base $q$ (i.e.\@ $Z[0],\dots,Z[n-1]$) under the constraint that $Z[i]$ must be outputted before $Z[i+1]$ and when $Z[i]$ is outputted, we are not allowed to use any knowledge of the digits $X[i+1],\dots,X[n-1]$ and $Y[i+1],\dots,Y[n-1]$. We can think of the digits of $X$ and $Y$ arriving one pair at a time, starting with the least significant pair of digits, and we output the corresponding digit of the product of the two numbers seen so far.

We also consider a variant of the online multiplication problem when one of the two input numbers, say $Y$\!, is known in advance. That is, all its digits are available at every stage of the algorithm and only the digits of $X$ arrive in an online fashion. In particular we will consider the case when $Y=K_{q,n}$ is fixed, where we define $K_{q,n}$ to be the largest number in $[q^n]$ such that the $i$th bit in the binary expansion of $K_{q,n}$ is $1$ if and only if $i$ is a power of two (starting with $i=0$ at the lower-order end). We can see that the binary expansion of $K_{q,n}$ is the reverse of $K_{(n\log_2 q)}$ in Equation~(\ref{eq:K}). We will prove the following result.
\begin{theorem}
    \label{thm:mult}
     For any positive integers $\delta$ and $n$ in the cell probe model with $w$ bits per cell, the expected running time of any algorithm solving the online multiplication problem on two $n$-digit random numbers $X,Y\in [q^n]$ is $\Omega(\frac{\delta}{w} n\log n)$, where $q=2^\delta$ is the base. The same bound holds even under full access to every digit of $Y$\!, and when $Y=K_{q,n}$ is fixed.
\end{theorem}

It suffices to prove the lower bound for the case when we have full access to every digit of $Y$; we could always ignore digits. We prove Theorem~\ref{thm:mult} using the same approach as for the online convolution problem. Here the $\update$-operation delivers a new digit of $X$, which is chosen uniformly at random from $[q]$, and outputs the corresponding digit of the product of $X$ and $Y$.

For $t_0,t_1,t_2\in[n]$ such that $t_0\leq t_1< t_2$, we write $\Xc$ to denote every digit of $X$ (in base~$q$) except for those at position $t_0$ through $t_1$. It is helpful to think of $\Xc$ as a vector of digits rather than a single number. We write $\Xfix$ to denote a fixed choice of $\Xc$. During the interval $[t_1+1,t_2]$, we output $Z[(t_1+1)\twodots t_2]$. The information transfer is defined as before, and Lemma~\ref{lem:entropy-upper} is replaced with the following lemma.

\begin{lemma}
    \label{lem:mult-upper}
    The entropy
    \begin{align*}
        H(Z[(t_1+1)\twodots t_2] \,&\mid\, \Xc = \Xfix) \leq\\
            &w + 2w\cdot \expect{|\IT{t_0,t_1,t_2}| \,\mid\, \Xc = \Xfix}\,.
    \end{align*}
\end{lemma}

\subsection{Retrorse numbers and the lower bound}

In Figure~\ref{fig:retrorse},\insertmultfigure the three numbers $X$, $Y$ and $Z=X\times Y$ are illustrated with some segments of their digits labelled $X'$, $Y'$ and $Z'$. Informally, we say that $Y'$ is \emph{retrorse} if $Z'$ depends heavily on $X'$. We have borrowed the term from Paterson, Fischer and Meyer~\cite{PFM:1974}, however, we give it a more precise meaning, formalised below.

Suppose $[t_0,t_1]$ and $[t_1+1,t_2]$ both have the same length $\ell$. For notational brevity, we write $X'$ to denote $X[t_0\twodots t_1]$, $Y'$ to denote $Y[0\twodots (2\ell-1)]$ and $Z'$ to denote $Z[(t_1+1)\twodots t_2]$ (see Figure~\ref{fig:retrorse}). We say that $Y'$ is \emph{retrorse} if for any fixed values of $t_0$, $\Xc$ (the digits of $X$ outside $[t_0, t_1]$) and $Y[2\ell\twodots (n-1)]$, each value of $Z'$ can arise from at most four different values of $X'$. That is to say there is at most a four-to-one mapping from possible values of $X'$ to possible values of $Z'$. We define $I_{Y,\ell}=\ell$ if $Y'$ is retrorse, otherwise $I_{Y,\ell}=0$. Note that $I_{Y,\ell}$ only depends on $Y$ and $\ell$. We will use $I_{Y,\ell}$ similarly to the recovery number $R_{\VB, \ell}$ from Section~\ref{sec:conv-recover} and replace Lemma~\ref{lem:entropy-lower} with Lemma~\ref{lem:mult-lower} below, which combined with Lemma~\ref{lem:mult-upper} gives us Corollary~\ref{cor:mult-IT}.

\begin{lemma}
    \label{lem:mult-lower}
    If the intervals $[t_0,t_1]$ and $[t_1+1,t_2]$ both have the same length~$\ell$, then the entropy
    \begin{equation*}
        H(Z[(t_1+1)\twodots t_2] \,\mid\, \Xc = \Xfix)
            \,\geq\, \frac{\delta I_{Y,\ell}}{4} - \frac{1}{2}\,.
    \end{equation*}
\end{lemma}
\begin{proof}
    The lemma is trivially true when $I_{Y,\ell}=0$, so suppose that $I_{Y,\ell}=\ell$. Then $Y[0\twodots (2\ell-1)]$ is retrorse, which implies that at most four distinct values of $X[t_0\twodots t_1]$ yield the same value of $Z[(t_1+1)\twodots t_2]$. There are $q^\ell$ possible values of $X[t_0\twodots t_1]$, each with the same probability, hence, from the definition of entropy,
    \begin{align*}
        H(Z[(t_1+1)\twodots t_2] \,&\mid\, \Xc = \Xfix) \,\geq\\
            &\frac{q^\ell}{4} \cdot \frac{1}{q^\ell} \cdot \log_2 \left(\frac{1}{4/q^\ell}\right)
                \,=\, \frac{\delta \ell}{4} - \frac{1}{2}\,.\tag*{\qedhere}
    \end{align*}
\end{proof}


\begin{corollary}
    \label{cor:mult-IT}
    For any fixed number $Y$\!, two intervals $[t_0,t_1]$ and $[t_1+1,t_2]$ of the same length $\ell$, and any algorithm solving the online multiplication problem on $X$ chosen uniformly at random from $[q^n]$,
    \begin{equation*}
        \expect{\,|\IT{t_0,t_1,t_2}|\,}\geq \frac{\delta I_{Y,\ell}}{8w} - 1\,.
    \end{equation*}
\end{corollary}

We take the same approach as in Section~\ref{sec:conv-lower} and use a lower-bound tree $\calT$ with $n$ leaves to obtain the next lemma. The proof is identical to the proof of Lemma~\ref{lem:conv-time}, only that we use Corollary~\ref{cor:mult-IT} instead of Corollary~\ref{cor:conv-IT}.

To avoid technicalities we will assume that $n$ and $\delta$ are both powers of two and we let the base $q=2^\delta$.

\begin{lemma}
    \label{lem:mult-time}
    For any fixed number $Y$ and any algorithm solving the online multiplication problem, the expected running time of the algorithm with the number $X$ chosen uniformly at random from $[q^n]$ is at least
    \begin{equation*}
        \frac{\delta}{8w}\sum_{v\in \calT} I_{Y,L(v)} - (n-1)\,.
    \end{equation*}
\end{lemma}

Before giving the proof of Theorem~\ref{thm:mult}, we bound the value of $I_{Y,\ell}$ for both a random number $Y$ and $Y=K_{q,n}$. In order to do so, we will use the following two results by Paterson, Fischer and Meyer~\cite{PFM:1974} which apply to binary numbers. The lemmas are stated in our notation, but the translation from the original notation of~\cite{PFM:1974} is straightforward.

\begin{lemma}[Lemma~1 of~\cite{PFM:1974}]
    \label{lem:PFMlemma1}
      For the base $q=2$ and fixed values of $t_0$, $\ell$, $n$ and $\Xc$ (where $t_1=t_0+\ell-1$), such that $\ell$ is a power of two, each value of $Z'$ can arise from at most two values of $X'$ when $Y=K_{2,n}$.
\end{lemma}

\begin{lemma}[Corollary of Lemma~5 in~\cite{PFM:1974}]
    \label{lem:PFMlemma5}
        For the base $q=2$ and fixed values of $t_0$, $\ell$, $n$ and $\Xc$ (where $t_1=t_0+\ell-1$), such that $\ell$ is a power of two, at least half of all possible values of $Y'$ have the property that each value of $Z'$ can arise from at most four different values of $X'$.
\end{lemma}

\begin{lemma}
    \label{lem:retrorse}
    If $\ell$ is a power of two, then for a random $Y \in [q^n]$, $\expectinline{I_{Y,\ell}}\geq \ell/2$, and for $Y=K_{q,n}$, $I_{Y,\ell}=\ell$.
\end{lemma}
\begin{proof}
    Suppose first that $Y=K_{q,n}$. Let $\ell$ be a power of two and $t_0$ a non-negative integer. We define $X'$, $Y'$ and $Z'$ as before (see Figure~\ref{fig:retrorse}). Instead of writing the numbers in base $q$, we consider their binary expansions, in which each digit is represented by $\delta=\log_2 q$ bits. In binary, we can write $X$, $Y$ and $Z$ as in Figure~\ref{fig:retrorse} if we replace $n$, $t_0$ and $\ell$ with $\delta n$, $\delta t_0$ and $\delta\ell$, respectively. Note that $\delta\ell$ is a power of two. Since $K_{q,n}=K_{2,\delta n}$, it follows immediately from Lemma~\ref{lem:PFMlemma1} that $Y'$ is retrorse and hence $I_{Y,\ell}=\ell$.

    Suppose now that $Y$ is chosen uniformly at random from $[q^n]$, hence $Y'$ is a random number in $[q^{2\ell}]$. From Lemma~\ref{lem:PFMlemma5} it follows that $Y'$ is retrorse with probability at least a half. Thus, $\expectinline{I_{Y,\ell}}\geq \ell/2$.
\end{proof}

\begin{proof}[Proof of Theorem~\ref{thm:mult}]
    We assume that $n$ is a power of two. Let $Y$ be a random number in $[q^n]$, either under the uniform distribution or the distribution in which $K_{q,n}$ has probability one and every other number has probability zero. A lower bound on the running time is obtained by taking the expectation of the bound in the statement of Lemma~\ref{lem:mult-time}. Using linearity of expectation and applying Lemma~\ref{lem:retrorse} and Fact~\ref{fac:leaves} finish the proof. Note from Lemma~\ref{lem:retrorse} that the expected value $\expectinline{I_{Y,\ell}}= \ell$ when $Y=K_{q,n}$.
\end{proof}

\section{Acknowledgements}
We are grateful to Mihai \Patrascu for suggesting the connection between online lower bounds and the recent cell-probe results for dynamic data structures and for very helpful discussions on the topic. We would also like to thank Kasper Green Larsen for the observation that our lower bounds are in fact tight within the cell-probe model.  MJ was supported by the EPSRC.

\printbibliography

\end{document}